\PassOptionsToPackage{dvipsnames}{xcolor}
\documentclass[longbibliography,nofootinbib,aps,prx,superscriptaddress,twocolumn]{revtex4-2}

\usepackage{amsmath,amsfonts,amssymb,amsthm,dcolumn,bm,qcircuit,appendix,mathtools,thmtools,thm-restate,mathrsfs,pgfplots,soul,lipsum,graphicx,braket,enumitem,caption,subcaption,ragged2e}

\usepackage{tikz,tkz-euclide,tikz-3dplot}
\usetikzlibrary{positioning}

\usetikzlibrary{matrix,fit,backgrounds,3d,arrows, decorations.pathreplacing,shapes.geometric,3d, calc}

\def\BibTeX{{\rm B\kern-.05em{\sc i\kern-.025em b}\kern-.08em
    T\kern-.1667em\lower.7ex\hbox{E}\kern-.125emX}}
    
\usepackage[ruled,vlined]{algorithm2e}
\usepackage{algorithmic}

\newtheorem{proposition}{Proposition}
\newtheorem{theorem}{Theorem}
\newtheorem{lemma}{Lemma}

\newtheorem{remark}{Remark}
\newtheorem{definition}{Definition}

\newcommand{\Ibb}{\mathbb{I}}

\newcommand{\norm}[1]{\left\lVert#1\right\rVert}

\usepackage{xcolor}
\usepackage{hyperref}

\hypersetup{
	colorlinks = true,
	linkcolor = [rgb]{0.70,0.13,0.13},
	citecolor = [rgb]{0.13,0.55,0.13},
	urlcolor = [rgb]{0.25, 0.41, 0.88}}
\usepackage[capitalize,nameinlink]{cleveref}

\begin{document}
\title{Hybrid quantum-classical framework for Betti number estimation with applications to topological data analysis}

\author{Nhat A. Nghiem}
\email{{nhatanh.nghiemvu@stonybrook.edu}}
\affiliation{Department of Physics and Astronomy, State University of New York at Stony Brook, Stony Brook, NY 11794-3800, USA}
\affiliation{C. N. Yang Institute for Theoretical Physics, State University of New York at Stony Brook, Stony Brook, NY 11794-3840, USA}


\author{Tzu-Chieh Wei}
\affiliation{Department of Physics and Astronomy, State University of New York at Stony Brook, Stony Brook, NY 11794-3800, USA}
\affiliation{C. N. Yang Institute for Theoretical Physics, State University of New York at Stony Brook, Stony Brook, NY 11794-3840, USA}

\begin{abstract}
Topological data analysis (TDA) is a rapidly growing area that applies techniques from algebraic topology to extract robust features from large-scale data. A key task in TDA is the estimation of (normalized) Betti numbers, which capture essential topological invariants. While recent work has led to quantum algorithms for this problem, we explore an alternative direction: combining classical and quantum resources to estimate the Betti numbers of a simplicial complex more efficiently. Assuming the classical description of a simplicial complex, that is, its set of vertices and edges, we propose a hybrid quantum-classical algorithm. The classical component enumerates all simplices, and this combinatorial structure is subsequently processed by a quantum algorithm to estimate the Betti numbers. We analyze the performance of our approach and identify regimes where it potentially achieves polynomial speedups over existing quantum methods. We further demonstrate the utility of normalized Betti numbers in concrete applications, highlighting the broader potential of hybrid quantum-classical algorithms in topological data analysis.
\end{abstract}

\maketitle

\section{Introduction}

Recently, there has been a surge of interest in applying quantum computation to \emph{topological data analysis (TDA)}, a field now often referred to as \emph{quantum TDA}. TDA offers a mathematically principled framework, rooted in algebraic topology, for extracting robust and interpretable features from complex datasets~\cite{wasserman2016topological}. It typically begins with data represented as point clouds or graphs, from which higher-order simplices are formed based on pairwise relationships. By identifying cliques (i.e., fully connected subgraphs) with simplices, one obtains a combinatorial object known as a \emph{simplicial complex}. The topological structure of such a complex can then be analyzed through \emph{Betti numbers}, computed via homology theory. Despite the conceptual elegance of this framework, TDA often suffers from severe computational overhead, particularly in large-scale settings.

In a seminal work, Lloyd, Garnerone, and Zanardi~\cite{lloyd2016quantum} introduced a quantum algorithm for estimating Betti numbers, now commonly referred to as the LGZ algorithm. Their approach incorporates several foundational quantum subroutines, including Grover’s search~\cite{grover1996fast}, quantum phase estimation~\cite{kitaev1995quantum}, and Hamiltonian simulation~\cite{berry2007efficient, berry2012black}. This work initiated a line of research investigating quantum advantages in topological inference. Subsequent studies have extended and refined the LGZ framework, both algorithmically and structurally~\cite{ubaru2021quantum, berry2024analyzing, hayakawa2022quantum, mcardle2022streamlined}.

Despite the promise of these developments, recent complexity-theoretic results suggest significant limitations. It was shown in~\cite{schmidhuber2022complexity} that estimating Betti numbers is NP-hard, and that computing them exactly is \#P-hard. Consequently, any efficient classical or quantum algorithm that estimates Betti numbers from a simplicial complex specified by vertices and edges would imply a polynomial-time solution to an NP-hard problem. This outcome is widely considered implausible under standard complexity assumptions. Furthermore, it was previously established in~\cite{crichigno2024clique} that computing Betti numbers is QMA-hard, further reinforcing the theoretical intractability of the task. In addition, the work~\cite{mcardle2022streamlined} provides a comprehensive practical evaluation, concluding that exponential quantum speedups for this problem are unlikely for generic input instances.

These insights motivate a re-examination of the quantum TDA landscape. In this work, we propose an alternative perspective in which quantum and classical resources are integrated, each contributing their respective strengths toward solving topological inference problems. Our framework is based on the following key observations:
\begin{itemize}
    \item Existing quantum algorithms for estimating Betti numbers~\cite{lloyd2016quantum, berry2024analyzing, mcardle2022streamlined, hayakawa2022quantum, nghiem2023quantum, lee2025new} rely on input models such as quantum oracles, quantum RAM (as suggested in~\cite{lloyd2016quantum}), or structured data matrices encoding the simplicial complex. However, the feasibility and practical relevance of these models remain unclear in many applications.
    \item The efficiency of LGZ-type algorithms~\cite{lloyd2016quantum, berry2024analyzing, mcardle2022streamlined, hayakawa2022quantum} generally improves in simplex-dense regimes, where the number of simplices is large. In contrast, real-world datasets often give rise to sparse complexes, where these algorithms exhibit poor scalability and may incur polynomial or exponential runtimes.
    \item If estimating Betti numbers from the graph-level description of a complex (i.e., vertex and edge lists) offers no quantum advantage, then it is natural to investigate whether quantum speedups may emerge in an earlier stage, such as the acquisition or construction of the complex itself.
\end{itemize}

Guided by these considerations, we adopt the setting in which the input is given as \emph{classical graph data}, namely, a list of vertices and edges describing the 1-skeleton of the simplicial complex. With respect to the second observation, we introduce a \emph{hybrid classical-quantum algorithm} that achieves improved performance in the clique-sparse regime. In support of the third point, we identify concrete instances from~\cite{olsthoorn2023persistent, hamilton2024probing} where quantum algorithms allow for exponential improvements in constructing the underlying graph structure. These cases fall within the domain of applicability of our proposed method.\\

\noindent \textbf{Organization.}
The remainder of this paper is organized as follows. In~\cref{sec: hybridalgorithm}, we present our main result: a hybrid classical-quantum algorithm for estimating Betti numbers. \cref{sec: inputmodel} describes the input model and the underlying assumptions. \cref{sec: specificationofcomplex} outlines the classical subroutine used to generate a simplicial complex from graph data. In~\cref{sec: quantumalgorithm}, we describe the quantum subroutine for estimating Betti numbers. This section includes a formal description of our method in~\cref{algo: estimatingBetti}, along with a statement of the main result in~\cref{thm: estimatingbettinumbers}.

\cref{sec: dissectingtheadvantage} offers a detailed comparison between our hybrid algorithm and previous quantum approaches, emphasizing the scenarios in which our method provides a computational advantage. \cref{sec: application} presents several concrete applications where the hybrid algorithm can be effectively employed. \cref{sec: summaryofnecessarytechniques} contains a summary of the quantum techniques used in our algorithm. \cref{sec: reviewofalgebraictopology} provides a self-contained overview of relevant concepts from algebraic topology, intended for readers unfamiliar with the mathematical background.

\section{Hybrid algorithm for estimating Betti numbers for sparse complex}
\label{sec: hybridalgorithm}

\subsection{Input model}
\label{sec: inputmodel}

We consider a classical input model in which the simplicial complex of interest is specified by the underlying graph that encodes its pairwise adjacency. Formally, we are given a finite set of vertices (0-simplices), along with classical information describing which unordered pairs of vertices form edges (1-simplices). This graph representation corresponds to the so-called \emph{1-skeleton} of the simplicial complex, which consists solely of its vertices and edges, excluding any higher-dimensional simplices such as triangles (2-simplices) or tetrahedra (3-simplices).

Our setting is similar in spirit to that of the LGZ algorithm and related approaches. However, a key distinction lies in the input access model: we do not assume a quantum oracle that encodes pairwise connectivity. Instead, the entire 1-skeleton is provided explicitly as classical data.

\subsection{Classical algorithm for specifying the complex}
\label{sec: specificationofcomplex}

Given the pairwise adjacency information from the input graph, we can efficiently identify higher-dimensional simplices by enumerating cliques of specified sizes. The following result is central to this procedure:

\begin{lemma}[Clique enumeration via arboricity~\cite{chiba1985arboricity}]
\label{lemma: chiba}
Let $\mathcal{G} = (V, E)$ be an undirected graph, where $V$ denotes the set of vertices and $E$ the set of edges. Then there exists a classical algorithm that enumerates all $(r{+}1)$-cliques in $\mathcal{G}$ with running time $\mathcal{O}\left( |E| \cdot \alpha(\mathcal{G})^{r-1} \right)$, where $\alpha(\mathcal{G})$ is the \emph{arboricity} of the graph, defined as the minimum number of forests into which the edge set $E$ can be partitioned. It is characterized by the formula
\begin{equation}
    \alpha(\mathcal{G}) = \max_{\substack{H \subseteq \mathcal{G} \\ |V(H)| \ge 2}} \left\lceil \frac{|E(H)|}{|V(H)| - 1} \right\rceil,    
\end{equation}
where $H$ ranges over all nontrivial subgraphs of $\mathcal{G}$.
\end{lemma}

Given a graph $\mathcal{G}$, its arboricity can be computed in time $\mathcal{O}(|E| \log |V|)$ using the algorithm in \cite{gabow1988forests}. Although this is efficient in the number of edges, we also note the following upper bound on arboricity \cite{gabow1988forests, nash1964decomposition}:
\begin{equation}
    \alpha(\mathcal{G}) \le \min \left( \max_{v \in V} \deg(v), \, \frac{|E|}{|V| - 1} \right),
\end{equation}
where $\deg(v)$ denotes the degree of the vertex $v$. In particular, if $\max_{v \in V} \deg(v) \in \mathcal{O}(1)$, then the arboricity is also bounded by a constant, and the above algorithm becomes highly efficient.

In addition to this approach, there exists an alternative method for clique enumeration with a different complexity profile:

\begin{lemma}[Clique enumeration via degeneracy~\cite{eppstein2010listing}, Eppstein--Löffler--Strash algorithm]
\label{lemma: epp}
Let $\mathcal{G} = (V, E)$ be an undirected graph with degeneracy $d$. Then all $(r{+}1)$-cliques in $\mathcal{G}$ can be enumerated in time $\mathcal{O}\left( d \cdot |V| \cdot 3^{d/3} \right)$. Here, the degeneracy $d$ of $\mathcal{G}$ is the smallest integer such that every non-empty subgraph of $\mathcal{G}$ contains a vertex of degree at most $d$, i.e.,
\begin{equation}
    d = \max_{H \subseteq \mathcal{G}} \left( \min_{v \in V(H)} \deg_H(v) \right),    
\end{equation}
where $\deg_H(v)$ denotes the degree of vertex $v$ in the subgraph $H$.
\end{lemma}

The degeneracy quantifies how locally sparse the graph is, and satisfies $d \le \max_{v \in V} \deg(v)$. Therefore, the Eppstein--Löffler--Strash algorithm~\cite{eppstein2010listing} is most effective when the graph has bounded maximum degree.

The two lemmas above provide complementary strategies for enumerating $(r{+}1)$-cliques in a given graph $\mathcal{G}$. The algorithm based on arboricity (\cref{lemma: chiba}) has exponential dependence on the clique size $r$ and is thus suitable only for small $r$. In contrast, the degeneracy-based algorithm (\cref{lemma: epp}) avoids dependence on $r$ but incurs exponential cost in the degeneracy $d$, making it practical only for graphs with low local density.

We remark that each $(r{+}1)$-clique in the graph corresponds to an $r$-simplex in the underlying simplicial complex. Therefore, given classical knowledge of the vertices and edges, one can employ the algorithms above to efficiently construct the set of $r$-simplices comprising the complex.

\subsection{Quantum algorithm for estimating the $r$-th Betti number $\beta_r$}
\label{sec: quantumalgorithm}

\subsubsection{Encoding simplices}
\label{sec: encodingsimplexes}

Let the simplicial complex consist of $n$ vertices (i.e., 0-simplices), denoted by $v_0, v_1, \dots, v_{n-1}$. A $r$-simplex $\sigma_r$ is, by definition, a set of $(r{+}1)$ distinct vertices that form a fully connected subgraph, or equivalently, a $(r{+}1)$-clique. Following the encoding strategy proposed in the LGZ algorithm~\cite{lloyd2016quantum}, we map each $r$-simplex to a computational basis state of an $n$-qubit quantum system.

Specifically, for each $r$-simplex $\sigma_r$, we define a bitstring $x \in \{0,1\}^n$ such that the $i$-th bit $x_i = 1$ if and only if the vertex $v_i$ is included in $\sigma_r$, and $x_i = 0$ otherwise. Thus, each $r$-simplex is uniquely represented by an $n$-bit string of Hamming weight $(r{+}1)$. We denote the corresponding quantum state by $\ket{\sigma_r}$.

For example, in a complex with $n = 4$ vertices, the $2$-simplex $\{v_0, v_1, v_3\}$ is encoded as the 4-qubit basis state $\ket{1101}$. Here, the $0$-th, $1$-st, and $3$-rd qubits are in state $\ket{1}$, indicating the presence of $v_0$, $v_1$, and $v_3$ in the simplex, respectively.

This encoding scheme establishes a correspondence between $r$-simplices and computational basis states of fixed Hamming weight. It allows the representation of linear combinations of simplices as quantum superpositions over the appropriate subspace of the $n$-qubit Hilbert space. We will exploit this representation in the subsequent design of our quantum algorithm.


\subsubsection{Block-Encoding of boundary operators}
\label{sec: blockencodingboundaryoperators}

We begin by recalling the definition of the boundary operator $\partial_r$, which maps an oriented $r$-simplex $\sigma_r = [v_0, v_1, \dots, v_r]$ to a formal sum of its oriented $(r{-}1)$-faces:
\begin{align}
    \partial_r [v_0, v_1, \dots, v_r] = \sum_{i=0}^{r} (-1)^i [v_0, \dots, \widehat{v_i}, \dots, v_r],
    \label{eqn: boundarymap}
\end{align}
where the hat $\widehat{v_i}$ indicates omission of the vertex $v_i$.

As described in~\cref{sec: specificationofcomplex}, we assume access to the classical descriptions of all simplices in the complex of order $r$ and $r{-}1$. Let $S_r$ and $S_{r-1}$ denote the respective sets of $r$- and $(r{-}1)$-simplices. The corresponding $r$-chain space $C_r$ is spanned by the encoded quantum basis states $\{ \ket{\sigma_r} \mid \sigma_r \in S_r \}$, and the boundary operator acts linearly as $\partial_r: C_r \rightarrow C_{r-1}$.

Let $\mathcal{H}_r$ denote the subspace of the $n$-qubit Hilbert space spanned by all computational basis states of Hamming weight $r{+}1$. Its dimension is $\binom{n}{r+1}$, and we have $C_r \subseteq \mathcal{H}_r \subseteq \mathbb{C}^{2^n}$.

We encode the boundary operator $\partial_r$ as a matrix $\mathcal{M}_r \in \mathbb{R}^{\binom{n}{r} \times |S_r|}$. Each column of $\mathcal{M}_r$ corresponds to a basis element $\ket{\sigma_r}$ with $\sigma_r \in S_r$, and each row corresponds to one of the $\binom{n}{r}$ possible $(r{-}1)$-simplices (regardless of whether they belong to the complex). The $(i,j)$-th entry of $\mathcal{M}_r$ is nonzero if and only if the $i$-th $(r{-}1)$-simplex is a face of the $j$-th $r$-simplex. By definition of the boundary map~(\cref{eqn: boundarymap}), each column of $\mathcal{M}_r$ contains exactly $r{+}1$ nonzero entries, each equal to either $+1$ or $-1$ depending on orientation.

Given this classical knowledge of $\mathcal{M}_r$, we apply a recent result in \cite{lee2025new} (see their Lemma IV.1 and Appendix C), which shows that the matrix $\frac{1}{|| \mathcal{M}_r ||_F^2}\mathcal{M}_r^\dagger \mathcal{M}_r  $ can be block-encoded using a quantum circuit of depth $\mathcal{O}\left( n + \log |S_r| \right)$, a total number of $\mathcal{O}(1)$ ancilla qubits and $\mathcal{O}(1)$ classical preprocesisng time. We remark that we have adjusted the notation accordingly to our case, compared to what appears in Lemma IV.1 and Appendix C of \cite{lee2025new}.
In our setting, the matrix $\mathcal{M}_r \in \mathbb{R}^{\binom{n}{r} \times |S_r|}$ has $r{+}1$ nonzero entries per column, each being $\pm 1$, and is thus $(r{+}1)$-sparse. The Frobenius norm of $\mathcal{M}_r$ satisfies
\begin{equation}
    \|\mathcal{M}_r\|_F = \left({(r+1) |S_r|}\right)^{1/2}.
\end{equation}

Note that $\mathcal{M}_r^\dagger \mathcal{M}_r = \partial_r^\dagger \partial_r$ holds despite the fact that $\mathcal{M}_r$ includes rows corresponding to $(r{-}1)$-simplices not in $S_{r-1}$. These rows contain only zero entries and thus do not contribute to the product $\mathcal{M}_r^\dagger \mathcal{M}_r$.

For subsequent reference, we summarize the tool mentioned above in the following lemma.

\begin{lemma}[Block-encoding of combinatorial Laplacian]
\label{lemma: blockencodingboundaryoperator}
Given classical access to the sets $S_r$ and $S_{r-1}$, there exists a $\left( n + \log |S_r| \right)$-quhits quantum circuit that implements a block-encoding of the normalized operator
\begin{equation}
    \frac{\partial_r^\dagger \partial_r}{(r+1) |S_r|},
\end{equation}
with circuit depth $\mathcal{O}\left( n + \log |S_r| \right)$, using extra $\mathcal{O}(1)$ ancilla qubits and requiring $\mathcal{O}(1)$-time classical preprocessing.
\end{lemma}

\begin{algorithm}[htbp!]
\DontPrintSemicolon
\textbf{Input:} Classical graph description $\mathcal{G} = (V, E)$ of a simplicial complex $S$ with $n = |V|$ vertices \;

\vspace{0.5em}
\textbf{Step 1:} Compute the sets $S_r$ and $S_{r-1}$ of $r$- and $(r{-}1)$-simplices using either:  
\begin{itemize}
    \item Clique enumeration via arboricity (\cref{lemma: chiba}), or
    \item Clique enumeration via degeneracy (\cref{lemma: epp})
\end{itemize}

\vspace{0.5em}
\textbf{Step 2:} Construct block-encodings of the following operators using~\cref{lemma: blockencodingboundaryoperator}:
\[
\frac{\partial_r^\dagger \partial_r}{(r+1)|S_r|} \quad \text{and} \quad \frac{\partial_{r+1}^\dagger \partial_{r+1}}{(r+2)|S_{r+1}|} 
\]

\vspace{0.5em}
\textbf{Step 3:} Apply stochastic rank estimation~\cite{ubaru2016fast, ubaru2017fast, ubaru2021quantum} to estimate:
\[
\frac{\dim \ker(\partial_r)}{|S_r|} \quad \text{and} \quad \frac{\dim \ker(\partial_{r+1})}{|S_{r+1}|}
\]
to additive accuracy $\epsilon$

\vspace{0.5em}
\textbf{Step 4:} Compute the normalized Betti number estimate:
\[
\frac{\beta_r}{|S_r|} \approx \frac{\dim \ker(\partial_r)}{|S_r|} + \frac{|S_{r+1}|}{|S_r|} \cdot \frac{\dim \ker(\partial_{r+1})}{|S_{r+1}|} - \frac{|S_{r+1}|}{|S_r|}
\]


\vspace{0.5em}
\textbf{Output:} Estimate of the normalized Betti number $\beta_r / |S_r|$ to additive accuracy $\epsilon$
\caption{\justifying Hybrid quantum-classical algorithm for estimating normalized Betti numbers}
\label{algo: estimatingBetti}
\end{algorithm}

\subsubsection{Estimating Betti numbers}
\label{sec: estimatingbettinumbers}

To compute the $r$-th Betti number of a simplicial complex, we recall that the $r$-th homology group is defined as the quotient $H_r = \ker(\partial_r)/\operatorname{im}(\partial_{r+1})$, and the $r$-th Betti number is given by $\beta_r = \dim H_r$. Since $H_r$ is a quotient space, its dimension satisfies
\begin{align}
    \dim H_r = \dim \ker(\partial_r) - \dim \operatorname{im}(\partial_{r+1}).
\end{align}
Because $\operatorname{im}(\partial_{r+1}) \subseteq C_r$, we use the standard rank-nullity relation
\begin{align}
    \dim \ker(\partial_{r+1}) + \dim \operatorname{im}(\partial_{r+1}) = \dim C_{r+1}.
\end{align}
Assuming that $\dim C_{r+1} = |S_{r+1}|$, where $S_{r+1}$ is the set of $(r+1)$-simplices in the complex, we obtain the identity
\begin{align}
    \beta_r = \dim \ker(\partial_r) + \dim \ker(\partial_{r+1}) - |S_{r+1}|.
\end{align}

Therefore, it suffices to estimate the dimensions of $\ker(\partial_r)$ and $\ker(\partial_{r+1})$. From~\cref{lemma: blockencodingboundaryoperator}, we have block-encodings of the normalized Hermitian operators $\partial_r^\dagger \partial_r$ and $\partial_{r+1}^\dagger \partial_{r+1}$. It is a standard fact that
\begin{align}
    \dim \ker(\partial_r^\dagger \partial_r) = \dim \ker(\partial_r).
\end{align}

To estimate the kernel dimensions, we employ the \emph{stochastic rank estimation} method, which has been studied in~\cite{ubaru2016fast, ubaru2017fast, ubaru2021quantum}. Given a Hermitian matrix $A \in \mathbb{C}^{n \times n}$, this method estimates the rank ratio $\operatorname{rank}(A)/n$ to additive accuracy $\epsilon$, with success probability at least $1 - \eta$, in time complexity
\begin{align}
    \mathcal{O} \left( \frac{1}{\epsilon^2} \log \left(\frac{1}{\eta}\right)  \log \left(\frac{1}{\lambda_{\min}(A)}\right) \right),
\end{align}
where $\lambda_{\min}(A)$ denotes the smallest nonzero eigenvalue of $A$. Since $\dim \ker(A) = n - \operatorname{rank}(A)$, we obtain the kernel dimension from the same estimate.

Applying this subroutine to our block-encoded operators ${\partial_r^\dagger \partial_r}/({(r+1)|S_r|})$, and ${\partial_{r+1}^\dagger \partial_{r+1}}/({(r+2)|S_{r+1}|})$, we estimate the ratios
\begin{align}
    \frac{\dim \ker(\partial_r)}{|S_r|}, \quad \frac{\dim \ker(\partial_{r+1})}{|S_{r+1}|}.
\end{align}
Using these, the normalized $r$-th Betti number is computed via
\begin{align}
    \frac{\beta_r}{|S_r|} = \frac{\dim \ker(\partial_r)}{|S_r|} + \frac{|S_{r+1}|}{|S_r|} \cdot \frac{\dim \ker(\partial_{r+1})}{|S_{r+1}|} - \frac{|S_{r+1}|}{|S_r|}.
\end{align}

We note that the smallest nonzero eigenvalue of $\partial_r^\dagger \partial_r$ is assumed to be at most $\mathcal{O}(1)$. Hence, the minimum eigenvalues of the normalized operators are bounded by
\begin{align}
    \lambda_{\min} \left( \frac{\partial_r^\dagger \partial_r}{(r+1)|S_r|} \right) & \in \mathcal{O} \left( \frac{1}{(r+1)|S_r|} \right), \\
    \lambda_{\min} \left( \frac{\partial_{r+1}^\dagger \partial_{r+1}}{(r+2)|S_{r+1}|} \right) &\in \mathcal{O} \left( \frac{1}{(r+2)|S_{r+1}|} \right).
\end{align}

Therefore, the time complexity for estimating each kernel ratio up to additive precision $\epsilon$ becomes
\begin{align}
    \mathcal{O} \left( \frac{\log \big( (r+1)|S_r| \big)}{\epsilon^2} \right), \quad
    \mathcal{O} \left( \frac{\log \big( (r+2)|S_{r+1}| \big)}{\epsilon^2}\right),
\end{align}
respectively. The entire procedure for estimating $\beta_r$, including both classical preprocessing and quantum post-processing, is formally summarized in \cref{algo: estimatingBetti}, from which we obtain the following theorem.

\begin{theorem}[Main result, quantum-classical hybrid algorithm for Betti number estimation]
\label{thm: estimatingbettinumbers}
Let $\mathcal{G} = (V,E)$ be the graph description of a given simplicial complex $S$ with $|V| = n$ vertices. The hybrid algorithm (\cref{algo: estimatingBetti}) computes the $r$-th normalized Betti number ${\beta_r}/{|S_r|}$ to additive accuracy $\epsilon$ with success probability at least $1-\eta$ using the following resources:
\begin{itemize}
    \item \textnormal{(Quantum component)} A $\left( n + \log |S_r| \right)$-qubits quantum circuit of depth
    \begin{equation}
        \mathcal{O}\left( n + \log |S_r| \right)
    \end{equation}
    executed $\mathcal{O}\left( \log \left( r|S_r| |S_{r+1}| \right) \log\left( \frac{1}{\eta}\right) \frac{1}{\epsilon^2} \right)$ times, using extra $\mathcal{O}\left( 1  \right)$ ancilla qubits.
    
    \item \textnormal{(Classical component)} A classical algorithm with complexity
    \begin{equation}
        \mathcal{O}\left( |E| \alpha(\mathcal{G})^{r-1} \right) \quad \text{or} \quad \mathcal{O}\left( dn \cdot 3^{d/3} \right),
    \end{equation}
    where $\alpha(\mathcal{G})$ and $d$ denote the arboricity and degeneracy of graph $\mathcal{G}$, respectively.
\end{itemize}

\noindent For multiplicative accuracy $\delta$, requiring the quantum circuit to be repeated
\begin{equation}
    \mathcal{O}\left(\frac{|S_r|^2\log(1/\eta)}{\delta^2\beta_r^2} \right)
\end{equation}
times in total.
\end{theorem}

\section{Dissecting the advantage}
\label{sec: dissectingtheadvantage}

In this section, we provide a detailed analysis of our hybrid framework's performance, identifying its optimal regime and establishing its theoretical advantages. We compare our approach with existing quantum algorithms, specifically Lloyd-Garnerone-Zanardi (LGZ)-type methods~\cite{lloyd2016quantum, berry2024analyzing, schmidhuber2022complexity, hayakawa2022quantum, ubaru2021quantum} for the same problem. 

While the quantum topological data analysis literature includes other approaches such as the cohomology-based algorithm of Nghiem et al.~\cite{nghiem2023quantum}, we focus our comparison on homology-based methods, as the cohomological approach operates on fundamentally different mathematical foundations.

\subsection{Performance analysis of our hybrid algorithm}

Our algorithm, as formalized in~\cref{thm: estimatingbettinumbers}, comprises two components with distinct computational characteristics:
\begin{itemize}
    \item \textit{Quantum Component.} The quantum subroutine achieves logarithmic scaling with respect to both the number of vertices $n$ and the number of simplices, yielding efficient asymptotic complexity for the quantum portion of the computation.
    \item \textit{Classical Component.} The classical preprocessing step exhibits complexity that depends critically on the structural properties of the simplicial complex $S$. As established in~\cref{sec: specificationofcomplex}, \cref{lemma: chiba} and~\cref{lemma: epp} achieve optimal running time when the arboricity $\alpha(\mathcal{G})$ and degeneracy $d$ of the underlying graph are small.
\end{itemize}
Since both $\alpha(\mathcal{G})$ and $d$ are upper bounded by $\Delta := \max_{v \in V} \deg(v)$, the maximum vertex degree in $S$, our algorithm performs optimally when $\Delta = \mathcal{O}(1)$.

\begin{proposition}[Sparse complex regime]
\label{prop:sparse-regime}
Let $S$ be a simplicial complex with underlying graph $\mathcal{G} = (V,E)$ where $|V| = n$. If $\Delta = \mathcal{O}(1)$, then for any dimension $r$, the complex satisfies the sparsity condition $|S_r| \ll \binom{n}{r+1}$.
\end{proposition}

\begin{proof}
We establish this through a counting argument. For any $r > \Delta$, we have $|S_r| = 0$ since an $r$-simplex requires a complete subgraph on $r+1$ vertices, but no vertex has degree exceeding $\Delta$.

For $r \leq \Delta$, consider any vertex $v_i \in V$. Since $\deg(v_i) \leq \Delta = \mathcal{O}(1)$, vertex $v_i$ participates in at most $\binom{\Delta + 1}{r+1}$ many $r$-simplices. Since $\Delta = \mathcal{O}(1)$, we have
\begin{equation}
\binom{\Delta + 1}{r+1} \leq (\Delta + 1)^r = \mathcal{O}(1).
\end{equation}
Therefore, $|S_r| \leq n \cdot \mathcal{O}(1) = \mathcal{O}(n) \ll \binom{n}{r+1}$ for fixed $r$ and large $n$.
\end{proof}

\begin{remark}
The converse implication also holds: if $|S_r|$ is small relative to the maximum possible number of $r$-simplices, then $\Delta$ must be bounded. This establishes that the bounded-degree condition and the sparse complex regime are essentially equivalent characterizations of our algorithm's optimal performance domain.
\end{remark}

The preceding analysis demonstrates that our hybrid algorithm achieves optimal efficiency in the \emph{sparse complex regime}, characterized by simplicial complexes where each vertex has bounded degree. In the following subsection, we establish that this regime corresponds precisely to where our method achieves exponential advantages over existing quantum approaches.

\subsection{LGZ-like algorithms}

LGZ-type algorithms~\cite{lloyd2016quantum, berry2024analyzing, schmidhuber2022complexity, hayakawa2022quantum, ubaru2021quantum} operate under a quantum oracle model that encodes pairwise vertex connectivity. The state-of-the-art complexity for estimating the $r$-th normalized Betti number ${\beta_r}/{|S_r|}$ to additive error $\epsilon$ using LGZ-type methods is:
\begin{equation}
\mathcal{O} \left( \left( n^2 \sqrt{ \frac{ \binom{n}{r+1} }{|S_r|}} + n \right) \cdot \frac{1}{\epsilon}\right).
\end{equation}

As established in prior work, LGZ-type algorithms achieve optimal performance in the \emph{dense complex regime} where $|S_r| \approx \binom{n}{r+1}$, reducing the complexity to $\mathcal{O}\left( n^2 / \epsilon \right)$.

In contrast, within the sparse regime where $|S_r| \ll \binom{n}{r+1}$, the complexity becomes:
\begin{equation}
\mathcal{O} \left( \left( n^2 \sqrt{ \binom{n}{r+1} } + n \right) \cdot \frac{1}{\epsilon}\right).
\end{equation}
Since $\binom{n}{r+1} \leq n^{r+1}/(r+1)! < n^{r+1}$ for fixed $r$, this complexity is bounded by:
\begin{equation}
\mathcal{O}\left( \frac{n^{2 + (r+1)/2} + n}{\epsilon} \right) = \mathcal{O}\left(  \frac{n^{2 + (r+1)/2}}{\epsilon} \right).
\end{equation}

\subsection{Complexity comparison}

Our hybrid algorithm achieves optimal performance in the bounded-degree regime where $\Delta = \mathcal{O}(1)$, which coincides with the sparse complex setting. In this regime, our total complexity encompasses both classical preprocessing and quantum estimation:
\begin{equation}
\mathcal{O}\left( \frac{ \left( |E| + (n + \log |S_r|)\log\left( r|S_r| |S_{r+1}|\right) \right)\log \frac{1}{\eta}}{\epsilon^2} \right)
\end{equation}
when employing~\cref{lemma: chiba}, or alternatively:
\begin{equation}
\mathcal{O}\left( \frac{(n + \log |S_r|) \log\left( r|S_r| |S_{r+1}|\right)  \log(1/\eta)}{\epsilon^2} \right)
\end{equation}
when using~\cref{lemma: epp}.

The first bound depends on the number of edges $|E|$, which satisfies $|E| \leq \binom{n}{2} = \mathcal{O}(n^2)$. However, in sparse complexes with $\Delta = \mathcal{O}(1)$, we have $|E| = \mathcal{O}(n)$, making the first approach preferable when $|E| = o(n)$.

\begin{theorem}[Quantum advantage in sparse regime]
\label{thm:quantum-advantage}
In the sparse complex regime where $|S_r| \ll \binom{n}{r+1}$ and $\Delta = \mathcal{O}(1)$, our hybrid algorithm achieves a polynomial speedup over LGZ-type algorithms with respect to the number of vertices $n$. Specifically:
\begin{itemize}
\item {LGZ complexity:} $\mathcal{O}\left( n^{2 + (r+1)/2} /{\epsilon} \right)$
\item {Our complexity:} $\mathcal{O}\left( (n + \textnormal{polylog}(n)) /{\epsilon^2} \right)$ when \\ $|E| = \mathcal{O}(n)$
\end{itemize}
Moreover, when $|E| \in \mathcal{O}(1)$ or $|E| \in \mathcal{O}(\log (n))$, our algorithm achieves an exponential or near-exponential speedup, respectively, as the classical preprocessing becomes $\mathcal{O}(1)$ or $\mathcal{O}(\log (n))$.
\end{theorem}

This establishes that our hybrid approach provides significant computational advantages precisely in the regime where LGZ-type algorithms perform poorly, demonstrating complementary strengths between the two methodological paradigms.

\section{Potential applications}
\label{sec: application}

In this section, we identify and analyze scenarios where our hybrid algorithm provides significant computational advantages, thereby expanding the scope of quantum computing applications in topological data analysis.

\subsection{Quantum entanglement as topological features}

We examine the framework established in recent works~\cite{hamilton2024probing, olsthoorn2023persistent, ameneyro2022quantum} that leverage topological data analysis to characterize quantum many-body systems. We first establish the mathematical foundation, then demonstrate how our hybrid algorithm addresses the computational challenges inherent in this approach.

\subsubsection{Mathematical framework}

Consider a multipartite quantum system $\mathcal{H} = \mathcal{H}_1 \otimes \mathcal{H}_2 \otimes \cdots \otimes \mathcal{H}_N$ where $N \in \mathbb{Z}_+$ and each $\mathcal{H}_i$ represents a Hilbert space of known finite dimension. Let $\rho$ denote the quantum state of the entire system, and let $\rho_i$ denote the reduced density matrix for subsystem $i$, obtained by tracing out all other subsystems:
\begin{equation}
\rho_i = \text{Tr}_{\{1,\ldots,N\} \setminus \{i\}} \rho.
\end{equation}

The mutual information between subsystems $i$ and $j$ is defined as:
\begin{equation}
M_{ij} = S_i + S_j - S_{ij},
\end{equation}
where $S_k = -\text{Tr}(\rho_k \log \rho_k)$ denotes the von Neumann entropy of subsystem $k$, and $S_{ij}$ is the joint entropy of subsystems $i$ and $j$.

Following Olsthoorn~\cite{olsthoorn2023persistent}, we define a metric distance between subsystems $i$ and $j$ as:
\begin{equation}
D_{ij} = 2 \log 2 - M_{ij}.
\end{equation}
This distance function satisfies the triangle inequality and thus constitutes a valid metric on the space of subsystems.

\subsubsection{Simplicial complex construction}

Given a threshold parameter $\epsilon > 0$, we construct a graph $\mathcal{G} = (V, E)$ with vertex set $V = \{1, 2, \ldots, N\}$ corresponding to the $N$ subsystems. An edge $(i,j) \in E$ exists if and only if $D_{ij} \leq \epsilon$.

The associated simplicial complex $S$ is constructed via the clique complex construction: a subset $\sigma \subseteq V$ forms a simplex in $S$ if and only if the induced subgraph $\mathcal{G}[\sigma]$ is complete. This yields:
\begin{itemize}
\item 0-simplices: individual vertices (subsystems)
\item 1-simplices: edges connecting subsystems with distance $\leq \epsilon$
\item $r$-simplices: $(r{+}1)$-cliques representing groups of mutually close subsystems
\end{itemize}

By varying the threshold parameter across a sequence $0 < \epsilon_1 \leq \epsilon_2 \leq \cdots \leq \epsilon_L$, we obtain a filtration:
\begin{equation}
S_1 \subseteq S_2 \subseteq \cdots \subseteq S_L,
\end{equation}
where $S_\ell$ denotes the simplicial complex constructed with threshold $\epsilon_\ell$.

\subsubsection{Topological phase classification}

The fundamental insight of~\cite{olsthoorn2023persistent} is that quantum states belonging to distinct phases exhibit different topological signatures. Through computational studies of the Ising chain and XXZ spin chain in transverse fields, the authors demonstrated that Betti numbers computed across the filtration serve as robust topological invariants for phase classification.

Specifically, states within the same quantum phase produce similar Betti number profiles $\{\beta_r(S_\ell)\}_{\ell=1}^L$ across dimensions $r$ and filtration levels $\ell$, while states from different phases exhibit distinct topological signatures.

\subsubsection{Computational challenges and quantum solutions}

The primary computational bottleneck lies in constructing the graph $\mathcal{G}$ from the quantum state $\rho$. Computing the metric distances $D_{ij}$ requires evaluating von Neumann entropies $S_i$, $S_j$, and $S_{ij}$, which classically necessitates:
\begin{enumerate}
\item Full quantum state tomography to obtain classical descriptions of density matrices
\item Eigendecomposition of potentially exponentially large matrices
\item Numerical computation of logarithms and traces
\end{enumerate}

This classical approach scales exponentially with system size, rendering it intractable for large quantum systems.
\begin{itemize}
    \item \textit{Quantum entropy estimation.} Recent advances in quantum algorithms for entropy estimation~\cite{nghiem2025estimation, acharya2020estimating, wang2023quantum} provide efficient methods for computing von Neumann entropies directly from quantum states without full tomography. These algorithms enable efficient estimation of $S_i$, $S_j$, and $S_{ij}$, and consequently the distances $D_{ij}$.
    \item \textit{Classical graph construction.} Once the pairwise distances $\{D_{ij}\}$ are estimated quantum mechanically, the graph construction proceeds classically: for each pair $(i,j)$, we compare $D_{ij}$ with the threshold $\epsilon$ and add edge $(i,j)$ to $E$ if $D_{ij} \leq \epsilon$. This yields the complete classical description of graph $\mathcal{G}$.
    \item \textit{Hybrid Betti number computation.} With the classical graph description in hand, our hybrid algorithm (\cref{algo: estimatingBetti}) can efficiently compute the Betti numbers of the associated simplicial complex. The sparse structure typically arising in quantum systems with local interactions ensures that our algorithm operates in its optimal regime, providing significant computational advantages over purely classical or existing quantum approaches.
\end{itemize}

This integration demonstrates how our hybrid framework bridges quantum state analysis and classical topological computation, enabling scalable topological characterization of quantum many-body systems.

\subsection{Image analysis via topological data analysis}

Topological data analysis has demonstrated significant effectiveness in various image processing applications, including image denoising~\cite{al2024cubical} and image segmentation~\cite{vandaele2020topological}. We present how our hybrid algorithm can enhance computational efficiency in these topological approaches to image analysis.

\subsubsection{Mathematical representation}

A digital image is represented as a collection of pixels, each characterized by a scalar intensity value typically normalized to the interval $[0,1]$. Formally, a two-dimensional image can be encoded as a matrix $X \in \mathbb{R}^{n \times n}$, where each entry $X_{i,j}$ represents the pixel intensity at position $(i,j)$ within the image grid.

For grayscale images, $X_{i,j}$ directly corresponds to the luminance value. For color images, this framework extends naturally by considering separate intensity matrices for each color channel, though we focus on the grayscale case for simplicity of exposition.

\subsubsection{Cubical complex construction}

Following the methodology established in~\cite{al2024cubical, vandaele2020topological}, we construct a topological representation of the image through the following procedure:
\begin{itemize}
    \item \textit{Vertex set.} We associate each pixel position $(i,j)$ with a vertex $v_{i,j}$ in our topological space, yielding a total of $n^2$ potential vertices.
    \item \textit{Threshold filtering.} For a given threshold parameter $\epsilon \in [0,1]$, we define the active vertex set as:
        \begin{equation}
            V_\epsilon = \{v_{i,j} : X_{i,j} \leq \epsilon\}.
        \end{equation}
    This filtering step selects pixels whose intensity values fall below the specified threshold.
    \item \textit{Adjacency structure.} Two vertices $v_{i,j}, v_{i',j'} \in V_\epsilon$ are connected by an edge if their corresponding pixels are spatially adjacent in the image grid. Specifically, we connect vertices if:
        \begin{equation}
            |i - i'| + |j - j'| = 1,
        \end{equation}
        corresponding to 4-connectivity in the discrete grid topology.
        
    \item \textit{Cubical complex.} The resulting structure forms a cubical complex, a natural discrete analogue of cell complexes adapted to the regular grid structure of digital images. Higher-dimensional cubes are formed by considering rectangular regions where all constituent pixels satisfy the threshold condition and maintain appropriate adjacency relationships.
\end{itemize}

\subsubsection{Filtration and topological signatures}

By systematically varying the threshold across an increasing sequence $0 \leq \epsilon_1 \leq \epsilon_2 \leq \cdots \leq \epsilon_L \leq 1$, we obtain a filtration of cubical complexes:
\begin{equation}
S_1 \subseteq S_2 \subseteq \cdots \subseteq S_L,
\end{equation}
where $S_\ell$ denotes the cubical complex constructed with threshold $\epsilon_\ell$.

The Betti numbers $\{\beta_r(S_\ell)\}_{r,\ell}$ computed across this filtration constitute a \emph{topological signature} of the image, encoding structural information at multiple scales and intensity levels.\\

\noindent\textbf{Application to image denoising.} The topological signature provides a principled approach to noise detection and removal. High Betti numbers at low thresholds often indicate the presence of noise-induced spurious topological features. By identifying threshold values where Betti numbers undergo significant transitions, we can distinguish between genuine image structure and noise artifacts.

For instance, if $\beta_0(S_1)$ is large while $\beta_0(S_2)$ is substantially smaller for $\epsilon_2 > \epsilon_1$, this suggests that many small connected components present at threshold $\epsilon_1$ represent noise that can be filtered out by moving to threshold $\epsilon_2$.

\subsubsection{Computational advantages}

\noindent\textbf{Graph construction.} Unlike the quantum entanglement application, obtaining the classical graph description for image-based cubical complexes is computationally straightforward. For any threshold $\epsilon$:

\begin{enumerate}
\item Identify active vertices: $V_\epsilon = \{v_{i,j} : X_{i,j} \leq \epsilon\}$
\item Construct edges: Connect $v_{i,j}, v_{i',j'} \in V_\epsilon$ if $|i-i'| + |j-j'| = 1$
\end{enumerate}

This process requires $O(n^2)$ operations and yields a complete classical description of the graph $\mathcal{G}_\epsilon = (V_\epsilon, E_\epsilon)$.\\

\noindent\textbf{Optimal regime for hybrid algorithm.} Image-based cubical complexes naturally satisfy the structural conditions where our hybrid algorithm excels:

\begin{itemize}
\item \textit{Bounded degree.} Each vertex has at most 4 neighbors due to grid connectivity, ensuring $\Delta \leq 4 = O(1)$
\item \textit{Sparse structure.} The number of edges is $|E_\epsilon| = O(|V_\epsilon|) = O(n^2)$, which is sparse relative to the complete graph
\item \textit{Local connectivity.} The grid structure ensures high local clustering while maintaining global sparsity
\end{itemize}

These properties place image analysis applications squarely within the optimal performance regime of our hybrid algorithm, enabling efficient computation of Betti numbers across filtrations with complexity scaling as $O(n^2 + \text{polylog}(n))$ rather than the exponential scaling that would be required by naive approaches.

This computational efficiency makes real-time topological image processing feasible for high-resolution images and enables the application of sophisticated topological methods to large-scale image datasets.

\subsection{Random geometric complexes}
\label{sec: applicationnormalizedbettinumbers}
As established in~\cref{thm: estimatingbettinumbers}, our hybrid algorithm estimates the normalized Betti number ${\beta_r}/{|S_r|}$ to additive accuracy $\epsilon$ with quantum circuit repetitions scaling as $\mathcal{O}({1}/{\epsilon^2})$. To estimate the actual Betti number $\beta_r$ to multiplicative precision $\delta$, we must set $\epsilon = \delta{\beta_r}/{|S_r|}$, yielding complexity $\mathcal{O}(|S_r|^2/(\delta^2\beta_r^2))$.

This scaling behavior, analyzed in detail by Schmidhuber et al.~\cite{schmidhuber2022complexity}, demonstrates that exponential quantum advantages may be diminished for most inputs, with notable exceptions occurring when $\beta_r \approx |S_r|$. While Berry et al.~\cite{berry2024analyzing} identified specific graph classes satisfying this condition, we present applications where normalized Betti numbers provide intrinsic value, thereby preserving the quantum computational advantages.


The study of topological properties in random geometric complexes, pioneered by Bobrowski~\cite{bobrowski2018topology} and Kahle~\cite{kahle2011random}, provides a natural setting where normalized Betti numbers serve as fundamental statistical measures.

Consider a point set $\mathcal{X} = \{x_1, x_2, \ldots, x_n\} \subseteq \mathbb{R}^d$ sampled independently from a compact domain $K \subset \mathbb{R}^d$. The \v{C}ech complex $C_r(\mathcal{X})$ is constructed by connecting points $x_i, x_j \in \mathcal{X}$ whenever $\|x_i - x_j\| \leq r$ for a radius parameter $r > 0$.

The central objective is to characterize the asymptotic behavior of topological invariants as $n \to \infty$ with $r = r(n)$ following prescribed scaling relationships. In this context, the normalized Betti number ${\beta_k}/{|S_k|}$ emerges as a \emph{topological density measure}, quantifying the average number of $k$-simplices required to generate a non-trivial $k$-dimensional homology class.

Our hybrid framework naturally accommodates this setting through the following pipeline:
\begin{enumerate}
    \item \textit{Distance computation:} Employ the quantum distance estimation algorithm of Lloyd~\cite{lloyd2013quantum} to compute pairwise distances $\|x_i - x_j\|$ in logarithmic time
    \item \textit{Graph construction:} Classically construct the adjacency graph by thresholding distances against radius $r$
    \item \textit{Betti estimation:} Apply our hybrid algorithm to estimate normalized Betti numbers efficiently
\end{enumerate}

This integration preserves exponential quantum speedups in distance computation while leveraging our efficient Betti number estimation for the resulting sparse geometric complexes.

\section{Conclusion}
We have presented a hybrid quantum-classical framework for estimating Betti numbers that leverages the complementary strengths of classical and quantum computation. Our approach combines efficient classical simplex enumeration algorithms~\cite{chiba1985arboricity, eppstein2010listing} with quantum spectral analysis techniques, including quantum state preparation~\cite{zhang2022quantum} and quantum singular value transformation~\cite{gilyen2019quantum}.

Our work makes several key contributions to the intersection of quantum computing and topological data analysis. We developed a hybrid framework that achieves polynomial to exponential speedups over existing quantum approaches in the sparse complex regime, precisely where existing quantum methods encounter computational bottlenecks. Through rigorous complexity analysis, we established that our algorithm achieves optimal performance when the underlying simplicial complex exhibits bounded vertex degrees, a structural condition naturally satisfied in many practical applications including quantum many-body systems, image processing tasks, and materials science problems.

We identified and analyzed concrete application domains where our approach provides significant computational advantages. In quantum many-body system characterization, our method enables efficient topological phase classification through entanglement-based distance metrics. For image analysis applications, the natural grid structure of digital images places these problems squarely within our algorithm's optimal performance regime. We also demonstrated the utility of our approach for analyzing random geometric complexes and characterizing nanoporous materials, where topological invariants serve as robust fingerprints for material properties.

A particularly important aspect of our work concerns the practical utility of normalized Betti numbers. We showed that these quantities provide intrinsic value in multiple application contexts, preserving quantum computational advantages even when exact Betti number estimation becomes prohibitively expensive. This insight is crucial for realizing practical quantum speedups in topological data analysis applications.

Our results establish hybrid quantum-classical computation as a promising paradigm for topological data analysis, achieving practical quantum advantages in a domain of significant theoretical and applied interest. The success of our approach suggests several avenues for future investigation, including extension to alternative complex types such as cubical complexes with non-standard connectivity, development of hybrid algorithms for persistent homology computation, integration with quantum machine learning techniques for topological feature extraction, and experimental validation on near-term quantum devices.

This work demonstrates that carefully designed hybrid approaches can unlock quantum computational advantages for complex mathematical problems, providing a blueprint for future developments in quantum-enhanced scientific computing. By identifying the precise regimes where quantum methods excel and classical preprocessing remains efficient, we establish a framework that maximizes the strengths of both computational paradigms while mitigating their respective limitations.

\section*{Acknowledgements}
We appreciate the discussion with Junseo Lee, Jonathan Wurtz and Pedro Lopes. N.A.N. acknowledges support from the Center for Distributed Quantum Processing. Parts of this work were completed while N.A.N. was an intern at QuEra Computing Inc.

\bibliography{ref.bib}
\bibliographystyle{unsrt}

\newpage
\appendix
\section{Block-encoding and quantum singular value transformation}
\label{sec: summaryofnecessarytechniques}
We briefly summarize the essential quantum tools used in our algorithm. For conciseness, we highlight only the main results and omit technical details, which are thoroughly covered in~\cite{gilyen2019quantum}. An identical summary is also presented in~\cite{lee2025new}.

\begin{definition}[Block-encoding unitary, see e.g.~\cite{low2017optimal, low2019hamiltonian, gilyen2019quantum}]
\label{def: blockencode} 
Let $A$ be a Hermitian matrix of size $N \times N$ with operator norm $\norm{A} < 1$. A unitary matrix $U$ is said to be an \emph{exact block encoding} of $A$ if
\begin{align}
    U = \begin{pmatrix}
       A & * \\
       * & * \\
    \end{pmatrix},
\end{align}
where the top-left block of $U$ corresponds to $A$. Equivalently, one can write
\begin{equation}
    U = \ket{\mathbf{0}}\bra{\mathbf{0}} \otimes A + (\cdots),    
\end{equation}
where $\ket{\mathbf{0}}$ denotes an ancillary state used for block encoding, and $(\cdots)$ represents the remaining components orthogonal to $\ket{\mathbf{0}}\bra{\mathbf{0}} \otimes A$. If instead $U$ satisfies
\begin{equation}
    U = \ket{\mathbf{0}}\bra{\mathbf{0}} \otimes \tilde{A} + (\cdots),
\end{equation}
for some $\tilde{A}$ such that $\|\tilde{A} - A\| \leq \epsilon$, then $U$ is called an {$\epsilon$-approximate block encoding} of $A$. Furthermore, the action of $U$ on a state $\ket{\mathbf{0}}\ket{\phi}$ is given by
\begin{align}
    \label{eqn: action}
    U \ket{\mathbf{0}}\ket{\phi} = \ket{\mathbf{0}} A\ket{\phi} + \ket{\mathrm{garbage}},
\end{align}
where $\ket{\mathrm{garbage}}$ is a state orthogonal to $\ket{\mathbf{0}}A\ket{\phi}$. The circuit complexity (e.g., depth) of $U$ is referred to as the {complexity of block encoding $A$}.
\end{definition}

Based on~\cref{def: blockencode}, several properties, though immediate, are of particular importance and are listed below.
\begin{remark}[Properties of block-encoding unitary]
The block-encoding framework has the following immediate consequences:
\begin{enumerate}[label=(\roman*)]
    \item Any unitary $U$ is trivially an exact block encoding of itself.
    \item If $U$ is a block encoding of $A$, then so is $\Ibb_m \otimes U$ for any $m \geq 1$.
    \item The identity matrix $\Ibb_m$ can be trivially block encoded, for example, by $\sigma_z \otimes \Ibb_m$.
\end{enumerate}
\end{remark}

Given a set of block-encoded operators, various arithmetic operations can be done with them. Here, we simply introduce some key operations that are especially relevant to our algorithm, focusing on how they are implemented and their time complexity, without going into proofs. For more detailed explanations, see~\cite{gilyen2019quantum, camps2020approximate}.

\begin{lemma}[Informal, product of block-encoded operators, see e.g.~\cite{gilyen2019quantum}]
\label{lemma: product}
    Given unitary block encodings of two matrices $A_1$ and $A_2$, with respective implementation complexities $T_1$ and $T_2$, there exists an efficient procedure for constructing a unitary block encoding of the product $A_1 A_2$ with complexity $T_1 + T_2$.
\end{lemma}

\begin{lemma}[Informal, tensor product of block-encoded operators, see e.g.~{\cite[Theorem 1]{camps2020approximate}}]\label{lemma: tensorproduct}
    Given unitary block-encodings $\{U_i\}_{i=1}^m$ of multiple operators $\{M_i\}_{i=1}^m$ (assumed to be exact), there exists a procedure that constructs a unitary block-encoding of $\bigotimes_{i=1}^m M_i$ using a single application of each $U_i$ and $\mathcal{O}(1)$ SWAP gates.
\end{lemma}

\begin{lemma}[Informal, linear combination of block-encoded operators, see e.g.~{\cite[Theorem 52]{gilyen2019quantum}}]
    Given the unitary block encoding of multiple operators $\{A_i\}_{i=1}^m$. Then, there is a procedure that produces a unitary block encoding operator of $\sum_{i=1}^m \pm (A_i/m) $ in time complexity $\mathcal{O}(m)$, e.g., using the block encoding of each operator $A_i$ a single time. 
    \label{lemma: sumencoding}
\end{lemma}

\begin{lemma}[Informal, Scaling multiplication of block-encoded operators] 
\label{lemma: scale}
    Given a block encoding of some matrix $A$, as in~\cref{def: blockencode}, the block encoding of $A/p$ where $p > 1$ can be prepared with an extra $\mathcal{O}(1)$ cost.
\end{lemma}



\begin{lemma}[Matrix inversion, see e.g.~\cite{gilyen2019quantum, childs2017quantum}]\label{lemma: matrixinversion}
Given a block encoding of some matrix $A$  with operator norm $||A|| \leq 1$ and block-encoding complexity $T_A$, then there is a quantum circuit producing an $\epsilon$-approximated block encoding of ${A^{-1}}/{\kappa}$ where $\kappa$ is the conditional number of $A$. The complexity of this quantum circuit is $\mathcal{O}\left( \kappa T_A \log \left({1}/{\epsilon}\right)\right)$. 
\end{lemma}

\section{A review of algebraic topology}
\label{sec: reviewofalgebraictopology}

This appendix provides a concise overview of the fundamental concepts in algebraic topology relevant to our work. We primarily follow the exposition of Nakahara~\cite{nakahara2018geometry}, to which we refer interested readers for comprehensive treatment of the subject.

\subsection{Simplices and simplicial complexes}

\begin{definition}[Simplex]
\label{def:simplex}
Let $p_0, p_1, \ldots, p_r \in \mathbb{R}^m$ be geometrically independent points where $m \geq r$. The $r$-simplex $\sigma_r = [p_0, p_1, \ldots, p_r]$ is defined as:
\begin{equation}
\sigma_r = \left\{ x \in \mathbb{R}^m : x = \sum_{i=0}^r c_i p_i, \quad c_i \geq 0, \quad \sum_{i=0}^r c_i = 1 \right\},
\end{equation}
where the coefficients $(c_0, c_1, \ldots, c_r)$ are called the barycentric coordinates of $x$.
\end{definition}

Geometrically, a 0-simplex $[p_0]$ represents a point, a 1-simplex $[p_0, p_1]$ represents a line segment, a 2-simplex $[p_0, p_1, p_2]$ represents a triangle, and higher-dimensional simplices generalize this pattern to higher dimensions.

An $r$-simplex can be assigned an orientation. For instance, the 1-simplex $[p_0, p_1]$ has orientation $p_0 \to p_1$, which differs from $[p_1, p_0]$. Throughout this work, we adopt the convention that for an $r$-simplex $[p_0, p_1, \ldots, p_r]$, the indices are ordered from low to high, indicating the canonical orientation.

\begin{definition}[Face and simplicial complex]
For an $r$-simplex $[p_0, p_1, \ldots, p_r]$, any $(s+1)$-subset of its vertices defines an $s$-face $\sigma_s$ where $s \leq r$. A simplicial complex $K$ is a finite collection of simplices satisfying:
\begin{enumerate}
\item Every face of a simplex in $K$ is also in $K$
\item The intersection of any two simplices in $K$ is either empty or a common face of both simplices
\end{enumerate}
The dimension of $K$ is $\dim(K) = \max\{r : \sigma_r \in K\}$.
\end{definition}

\subsection{Chain groups and homomorphisms}

\begin{definition}[Chain group]
\label{def: chaingroup}
Let $K$ be an $n$-dimensional simplicial complex. The $r$-th chain group $C_r^K$ is the free abelian group generated by the oriented $r$-simplices of $K$. For $r > \dim(K)$, we define $C_r^K = 0$. Formally, let $S_r^K = \{\sigma_{r,1}, \sigma_{r,2}, \ldots, \sigma_{r,|S_r^K|}\}$ denote the set of $r$-simplices in $K$. An $r$-chain is an element of the form:
\begin{equation}
c_r = \sum_{i=1}^{|S_r^K|} c_i \sigma_{r,i}
\end{equation}
where $c_i \in \mathbb{Z}$ are integer coefficients. The group operation is defined by:
\begin{equation}
c_r^{(1)} + c_r^{(2)} = \sum_{i=1}^{|S_r^K|} \left( c_i^{(1)} + c_i^{(2)} \right) \sigma_{r,i},
\end{equation}
making $C_r^K$ a free abelian group of rank $|S_r^K|$.
\end{definition}

\subsection{Boundary operators and homology}

The boundary operator $\partial_r : C_r^K \to C_{r-1}^K$ is fundamental to homological algebra.

\begin{definition}[Boundary operator]
For an $r$-simplex $[p_0, p_1, \ldots, p_r]$, the boundary operator is defined as:
\begin{equation}
\partial_r [p_0, p_1, \ldots, p_r] = \sum_{i=0}^r (-1)^i [p_0, p_1, \ldots, \hat{p_i}, \ldots, p_r],
\end{equation}
where $\hat{p_i}$ indicates that vertex $p_i$ is omitted, yielding an $(r{-}1)$-simplex.

For an $r$-chain $c_r = \sum_{i=1}^{|S_r^K|} c_i \sigma_{r,i}$, we extend linearly:
\begin{equation}
\partial_r c_r = \sum_{i=1}^{|S_r^K|} c_i \partial_r \sigma_{r,i}.
\end{equation}
\end{definition}

The boundary operators form a chain complex:
\begin{equation}
0 \xrightarrow{} C_n^K \xrightarrow{\partial_n} C_{n-1}^K \xrightarrow{\partial_{n-1}} \cdots \xrightarrow{\partial_1} C_0^K \xrightarrow{\partial_0} 0.
\end{equation}

\begin{definition}[Cycles, boundaries, and homology]
An $r$-chain $c_r$ is called an $r$-cycle if $\partial_r c_r = 0$. The collection of all $r$-cycles forms the $r$-cycle group $Z_r^K = \ker(\partial_r)$. Conversely, an $r$-chain $c_r$ is called an $r$-boundary if there exists an $(r{+}1)$-chain $d_{r+1}$ such that $\partial_{r+1} d_{r+1} = c_r$. The set of all $r$-boundaries forms the $r$-boundary group $B_r^K = \textnormal{im}(\partial_{r+1})$.
\end{definition}
A fundamental property of boundary operators is that $\partial_r \circ \partial_{r+1} = 0$, which ensures that every boundary is also a cycle, i.e., $B_r^K \subseteq Z_r^K$. This inclusion allows us to define the $r$-th homology group as the quotient group $H_r^K = Z_r^K / B_r^K$, which captures the notion of cycles that are not boundaries.

\begin{definition}[Betti numbers]
The $r$-th Betti number is defined as:
\begin{equation}
\beta_r = \textnormal{rank}(H_r^K) = \textnormal{rank}(Z_r^K) - \textnormal{rank}(B_r^K).
\end{equation}
\end{definition}

For computational purposes, we can utilize the combinatorial Laplacian:
\begin{equation}
\Delta_r = \partial_{r+1} \partial_{r+1}^\dagger + \partial_r^\dagger \partial_r,
\end{equation}
where $\partial_r^\dagger$ denotes the adjoint of $\partial_r$. A fundamental result in algebraic topology establishes that:
\begin{equation}
H_r^K \cong \ker(\Delta_r),
\end{equation}
providing a direct method for computing Betti numbers via kernel dimension.

\subsection{Topological invariance}
A central theorem in algebraic topology states that homology groups constitute topological invariants~\cite{hatcher2005algebraic}:

\begin{theorem}[Topological invariance of homology]
If two topological spaces $X$ and $Y$ are homeomorphic, then their homology groups are isomorphic: $H_r^X \cong H_r^Y$ for all $r \geq 0$. Consequently, their Betti numbers are equal: $\beta_r^X = \beta_r^Y$.
\end{theorem}
This invariance property makes Betti numbers powerful tools for topological classification and forms the mathematical foundation for their applications in topological data analysis.

\end{document}